\documentclass[a4paper,envcountsame,envcountsect]{llncs}

\usepackage{hyperref}
\usepackage{textgreek}
\usepackage{amsmath}
\usepackage{amssymb}
\usepackage{graphicx}
\usepackage{mathtools}

\makeatletter
\def\@seccntformat#1{\@ifundefined{#1@cntformat}%
   {\csname the#1\endcsname\quad}  
   {\csname #1@cntformat\endcsname}
}
\let\oldappendix\appendix 
\renewcommand\appendix{%
    \oldappendix
    \newcommand{\section@cntformat}{\appendixname~\thesection\quad}
}
\makeatother

\usepackage{bbm}
\usepackage[greek,english]{babel}
\usepackage[utf8]{inputenc}

\usepackage{mathpartir}
\usepackage{url}
\newcommand{\keywords}[1]{\par\addvspace\baselineskip
\noindent\keywordname\enspace\ignorespaces#1}

\DeclareUnicodeCharacter{2200}{\ensuremath{\forall}}
\DeclareUnicodeCharacter{2605}{\ensuremath{\star}}
\DeclareUnicodeCharacter{2794}{\ensuremath{\rightarrow}}
\DeclareUnicodeCharacter{2243}{\ensuremath{\simeq}}
\DeclareUnicodeCharacter{03A0}{\ensuremath{\mathrm{\Pi}}}
\DeclareUnicodeCharacter{03BB}{\ensuremath{\mathrm{\lambda}}}
\DeclareUnicodeCharacter{25C2}{\ensuremath{\blacktriangleleft}}
\DeclareUnicodeCharacter{27BE}{\ensuremath{\Rightarrow}}
\DeclareUnicodeCharacter{2218}{\ensuremath{\circ}}

\newtheorem{notation}{Notation}
\newtheorem{convention}{Convention}
\newtheorem{assumption}{Assumption}

\newcommand{\eqed}[0]{\tag*{\qed}}

\newcommand{\refsec}[1]{Section \ref{sec:#1}}
\newcommand{\labsec}[1]{\label{sec:#1}}
\newcommand{\reffig}[1]{Figure \ref{fig:#1}}
\newcommand{\labfig}[1]{\label{fig:#1}}
\newcommand{\refrem}[1]{Remark \ref{rem:#1}}
\newcommand{\labrem}[1]{\label{rem:#1}}

\newcommand{\labasm}[1]{\label{asm:#1}}

\newcommand{\refthm}[1]{Theorem \ref{thm:#1}}
\newcommand{\labthm}[1]{\label{thm:#1}}
\newcommand{\reflem}[1]{Lemma \ref{lem:#1}}
\newcommand{\lablem}[1]{\label{lem:#1}}

\newcommand{\nega}[0]{\ensuremath{{\text -}}}
\newcommand{\earg}[1]{\,\,#1}
\newcommand{\iarg}[1]{\ensuremath{\,\nega #1}}
\newcommand{\targ}[1]{\ensuremath{\cdot #1}}

\newcommand{\subtp}[0]{\leq}

\newcommand{\cdle}[0]{\ensuremath{\iota \lambda P2}~}
\newcommand{\erase}[1]{\ensuremath{\lvert #1 \rvert}}
\newcommand{\Erase}[1]{\ensuremath{\Big\lvert #1 \Big\rvert}}
\newcommand{\rewrite}[2]{\ensuremath{\rho ~ #1 ~\nega~ #2}}
\newcommand{\pair}[2]{\ensuremath{[ #1 ~,~ #2]}}
\newcommand{\refl}[1]{\ensuremath{\beta\{#1\}}}

\newcommand{\fun}[1]{\lambda #1 .~}
\newcommand{\all}[1]{\Lambda #1 .~}
\newcommand{\by}[1]{\text{#1}}
\newcommand{\name}[1]{\textrm{#1}}

\begin{document}

\mainmatter
\title{Zero-Cost Coercions \\for Program and Proof Reuse}
\author{Larry Diehl and Aaron Stump}
\institute{University of Iowa\\
\email{\{larry-diehl,aaron-stump\}@uiowa.edu}}
\maketitle

\begin{abstract}
  We introduce the notion of \textit{identity coercions} between
  non-indexed and indexed variants of inductive datatypes, such as lists and
  vectors. An identity coercion translates one type to another such that
  the coercion function definitionally reduces to the identity
  function. This allows us to reuse vector programs to derive list
  programs (and vice versa), without any runtime cost. This also
  allows us to reuse vector proofs to derive list proofs (and vice
  versa), without the cost of equational reasoning proof obligations.
  Our work is formalized in Cedille, a dependently typed programming
  language based on a type-annotated Curry-style type theory with
  implicit (or, erased) products (or, dependent functions), and relies
  crucially on \textit{erasure} to introduce definitional
  equalities between underlying untyped terms.

\keywords{
Dependent types; coercion; reuse; implicit products; cedille.
}

\end{abstract}

\section{Introduction}
\labsec{intro}

In dependently typed languages
(such as Agda~\cite{lang:agda}, Coq~\cite{lang:coq},
Idris~\cite{lang:idris}, or Lean~\cite{lang:lean})
it is common to define traditional algebraic datatypes,
as well as more refined \textit{indexed} variants
of algebraic datatypes, where the values of the indexed type are a
restriction of the values of the original algebraic type to particular indices.
An example of two such datatypes are lists and vectors, vectors being
lists indexed by their length.

To prevent code duplication, a programmer may want to define a
function over lists by reusing a function over vectors (or vice versa),
which we refer to as \textit{program reuse}.
For example, we can derive list append (\texttt{appendL}) by reusing
vector append (\texttt{appendV}) as follows:

\begin{verbatim}
appendL : ∀ A : ★ . List A ➔ List A ➔ List A
appendL = λ xs ys . v2l (appendV (l2v xs) (l2v ys))
\end{verbatim}

This is achieved by coercing the list arguments to vectors
(via \texttt{l2v}), passing
them to the reused function \texttt{appendV}, and coercing
the resulting vector to a list (via \texttt{v2l}).
Unfortunately, this has the drawback of linear-time coercions back and
forth between lists and vectors (via \texttt{l2v} and \texttt{v2l})
when we \textit{run} our code.

A programmer may also want to prevent code duplication by defining a
proof of a property about list functions (defined by reuse) in terms of a proof of a
property about vector functions (or vice versa), which we refer to as
\textit{proof reuse}. For example, we may want to derive associativity
of list append (\texttt{appendAssocL}) in terms of associativity
of vector append (\texttt{appendAssocV}) as follows:

\begin{verbatim}
appendAssocL : ∀ A : ★ . Π xs ys zs : List A .
  appendL (appendL xs ys) zs ≃ appendL xs (appendL ys zs)  
appendAssocL = λ xs ys zs . cong v2l 
  (appendAssocV (l2v xs) (l2v ys) (l2v zs))
\end{verbatim}

Unfortunately, reusing the proof of \texttt{appendAssocV} by casting
our arguments to lists (via \texttt{v2l}),
and by congruence (\texttt{cong}) applied to the
cast back to vectors (via \texttt{l2v}), is not
enough to get the proof above to type check. We must additionally
perform equational reasoning, by appealing to the identity laws
established by an isomorphism between lists and vectors. In other
words, the proof would need to rewrite occurrences of
(\texttt{(v2l ∘ l2v) xs}) and (\texttt{(l2v ∘ v2l) xs})
to \texttt{xs} in the appropriate places, which may only appear after
previous rewrites and $\beta$-reductions.

We show that in a type-annotated implementation of a Curry-style type
theory, coercions that definitionally reduce to the identity function
($\lambda x. x$)
are \textit{derivable}, and we call them ``\textit{identity coercions}''.
Identity coercions enable \textit{zero-cost program reuse}, avoiding
runtime overhead, and \textit{zero-cost proof reuse}, avoiding
equational reasoning overhead
(making \texttt{appendAssocL} above well-typed).

\subsection{The Setting}

In a Curry-style type theory with implicit
products (such as ICC~\cite{miquel:implicit}),
an untyped Church-encoded vector can be assigned
the vector type (\texttt{Vec}),
but also the list type (\texttt{List}).
This is possible because the types share the same class of untyped
terms, and because vectors are a subtype of lists in ICC
($\name{Vec} \earg A \earg n \subtp \name{List} \earg A$).

A type-annotated version of a Curry-style
calculus with implicit products
(such as ICC*~\cite{barras:implicit} and
\cdle\cite{stump17b}) adds typing information to terms, but compares
\textit{erased} terms (removing type annotations, implicit type
applications, etc.) during conversion
($\erase{t} =_{\alpha\beta\eta} \erase{t'}$).
The extra type annotations on terms allows them
to be algorithmically type checked, making type-annotated versions of
Curry-style calculi suitable as the basis of programming languages.

This paper is formalized in Cedille, a dependently typed programming
language based on \cdle.\footnote{\raggedright{A pre-release of Cedille for
evaluating our formalization is available here:\\
\url{http://cs.uiowa.edu/~astump/cedille-prerelease.zip}}\\
  The Cedille code accompanying this paper is here:
  \url{https://github.com/larrytheliquid/zero-cost-coercions}}
In a type-annotated setting, a vector cannot be used in
the place of a list, as they have distinct types, despite the fact that
their \textit{erased} untyped values are equal.
Nonetheless, Barras and Bernardo~\cite{barras:implicit} demonstrate
(in ICC*) that it is possible to write an \textit{identity coercion} from
Church-encoded vectors to Church-encoded lists, which can be thought
of as a checkable term witness of the subtyping relationship:
$\name{Vec} \earg A \earg n \subtp \name{List} \earg
A$.\footnote{Barras and Bernardo did not name their technique, which
  we are calling ``identity coercion''.
  }

\subsection{Contributions}

In \cdle, Stump~\cite{stump17b} adds a
dependent intersection type~\cite{kopylov03}
and a heterogeneous equality type~\cite{mcbride00}
to a type-annotated Curry-style calculus
with implicit products, allowing \textit{inductive types} (i.e. those
supporting an induction principle) to be derived,
but whose erased terms are untyped Church-encodings.
Working in Cedille (based on \cdle), our contributions are:
\begin{enumerate}
\item{\textbf{Extending}}
  the \textit{non-dependent} identity coercion from \textit{Church-encoded}
  vectors to lists, to an identity coercion from \textit{inductive}
  vectors to lists (\texttt{v2l} in \refsec{vecreuse:v2l}). By working
  with inductive types, we can write proofs by induction but still
  support identity coercion.
\item{\textbf{Introducing}}
  the \textit{dependent} identity coercion from inductive
  lists to vectors (\texttt{l2v} in \refsec{vecreuse:l2v}). This is a witness of the
  dependent subtyping relationship:
  $(xs : \name{List}  \earg A) \subtp \name{Vec} \earg A \earg (\name{length} \earg xs)$.
  Because the length of the output vector depends on the input list, the dependent
  identity coercion \texttt{l2v} cannot be written using Church-encoded
  datatypes, which do not have
  induction principles~\cite{geuvers01}.
\item{\textbf{Introducing}} the identity coercion from inductive
  vectors to length-{\linebreak}constrained lists (\texttt{v2u} in
  \refsec{listreuse:v2u}). This allows vectors to be coerced to lists,
  while ``remembering'' the constraint that the length of the output
  list should be the length of the input vector index.
\item{\textbf{Introducing}}
  a functorial \textit{map} for inductive lists
  (\texttt{mapL} in \refsec{nestreuse:mapl}), whose
  \textit{partial application} to an identity coercion results in an
  identity coercion.
\end{enumerate}

After reviewing how to derive inductive datatypes in Cedille
(\refsec{inductive}),
we show how to reuse a vector program
(\texttt{appendV}) and proof (\texttt{appendAssocV})
to define a list program (\texttt{appendL}) and proof
(\texttt{appendAssocL}) in \refsec{vecreuse},
show how to reuse a list program
(\texttt{appendL}) and proof (\texttt{appendAssocL})
to define a vector program (\texttt{appendV}) and proof
(\texttt{appendAssocV}) in \refsec{listreuse},
show how to reuse a \textit{nested} list program
(\texttt{concatL}) and proof (\texttt{concatDistAppendL})
to define a nested vector program (\texttt{concatV})
and proof (\texttt{concatDistAppendV})
in \refsec{nestreuse},
discuss related work in \refsec{related},
and conclude in \refsec{conc}.
We reiterate that all of our instances of program and proof reuse
are \textit{zero-cost}, as they are implemented in terms of
identity coercions.

\begin{remark}
  All lemmas and theorems in this paper are trivial consequences of
  definitional equality ($\erase{t} =_{\alpha\beta\eta} \erase{t'}$).
  Nonetheless, we prove them by hand to aid the reader in
  understanding why terms erase the way they do, and in particular how
  our carefully crafted identity coercions indeed erase to the
  identity function (up to $\alpha\beta\eta$-equality).
\end{remark}

\section{Background: Deriving Inductive Types}
\labsec{inductive}

In this section we review how to derive inductive types in
Cedille~\cite{stump17b}, whose erased terms are untyped
Church-encodings.
An inductive datatype is defined as the dependent intersection
of 3 components:
\begin{enumerate}
\item The \textit{Church-encoding} of the datatype.
\item The \textit{unary parametricity theorem} of the Church-encoding.
\item The \textit{reflection theorem} of the Church-encoding.
\end{enumerate}

\subsection{Church-Encoding}
\labsec{ind:church}

The first component (\texttt{VecC}) is the Church-encoded vector type,
where the impredicatively quantified
\texttt{X} is a family of types indexed by the natural numbers.

\begin{convention}
We include ``\texttt{C}'' in
the suffix of an identifier to indicate that it relates to
a Church-encoded datatype.
\end{convention}

\begin{verbatim}
VecC ◂ ★ ➔ Nat ➔ ★ = λ A : ★ . λ n : Nat . 
  ∀ X : Nat ➔ ★ .
  X zero ➔
  (∀ n : Nat . A ➔ X n ➔ X (suc n)) ➔
  X n .
\end{verbatim}

Implicit products are introduced by the $\forall$ quantifier, and
represent erased dependent function arguments.
Implicit products can be used for type arguments (e.g. the type family
\texttt{X}), but also for value arguments
(e.g. \texttt{n} in the cons case, used for indexing).

\paragraph{Constructors}

Next, we define the Church-encoded vector constructors \texttt{nilCV}
and \texttt{consCV}.

\begin{convention}
We suffix an identifier with ``\texttt{V}''
to indicate that it relates to vectors.
\end{convention}

\begin{verbatim}
nilCV ◂ ∀ A : ★ . VecC · A zero =
  Λ A . Λ X . λ cN . λ cC . cN .

consCV ◂ ∀ A : ★ . ∀ n : Nat . A ➔ VecC · A n ➔ VecC · A (suc n) =
  Λ A . Λ n . λ x . λ xs . Λ X . λ cN . λ cC . 
  cC -n x (xs · X cN cC) .
\end{verbatim}

\begin{remark}
A full definition of term erasure for our base theory \cdle can be
found in Figure 6 of Stump'17~\cite{stump17b}.
\end{remark}

All of the implicit abstractions ($\Lambda$) are erased, as are
implicit applications (prefixed by minus, e.g. \texttt{-n}),
and type applications
(prefixed by a centered dot, e.g. \texttt{· X}).\footnote{While
  $\forall$ and $\Lambda$ quantify over and
  introduce (respectively) both type and value arguments,
  minus ($\nega$) is special syntax for implicit (value)
  applications, while center dot ($\cdot$) is special syntax for type
  applications.
  }
We can
verify that erasing the type-annotated \texttt{nilCV} and
\texttt{consCV} results in the untyped Church-encodings of nil and
cons, respectively:

\begin{lemma}
\erase{\name{nilCV}} is the Church-encoding of nil.
\lablem{nilcv}
\end{lemma}

\begin{proof}
{\small
\abovedisplayskip=-\baselineskip
\begin{align*}
  &=_\delta \erase{
    \all{A, X} \fun{c_n, c_c} c_n
  }
  && \by{Erase implicit abstractions.}
  \\
  &=~ \fun{c_n, c_c} c_n
  \eqed
\end{align*}}
\end{proof}

\begin{lemma}
\erase{\name{consCV}} is the Church-encoding of cons.
\lablem{conscv}
\end{lemma}

\begin{proof}
{\small
\begin{align*}
  &=_\delta \erase{
    \all{A, n} \fun{x, xs}
    \all{X} \fun{c_n, c_c}
    c_c \iarg{n} \earg{x} \earg (xs \targ{X} \earg{c_n} \earg{c_c})
  }
  && \by{Erase implicit abstractions.}
  \\
  &=~ \fun{x, xs, c_n, c_c} \erase{
    c_c \iarg{n} \earg{x} \earg (xs \targ{X} \earg{c_n} \earg{c_c})
  }
  && \by{Erase implicit applications.}
  \\
  &= \fun{x, xs, c_n, c_c} c_c \earg{x} \earg (xs \earg{c_n} \earg{c_c})
  \eqed
\end{align*}
}
\end{proof}

\subsection{Unary Parametricity Theorem}
\labsec{ind:param}

The second component (\texttt{VecP}) is the unary parametricity
predicate on Church-encoded vectors (\texttt{VecC}). It takes 4 types
of abstract arguments, described below, and can be understood as an
abstract version of an \textit{eliminator} (i.e. an induction
principle in type theory):

\begin{enumerate}
\item An abstract \textit{return type} (\texttt{X}).
\item An abstract \textit{motive} (\texttt{P}, an abstract predicate
  over \texttt{X}).
\item Abstract Church-encoded \textit{constructors}
  (\texttt{cN} for nil and \texttt{cC} for cons).
\item Abstract parametricity \textit{branches}
  (\texttt{pN} for nil branch and \texttt{pC} for cons branch).
\end{enumerate}

The reader may wish to compare the
vector parametricity theorem type
below (\texttt{VecP}), which has abstract versions of arguments, with the type of
the eliminator \texttt{elimVec} in \refsec{ind:ind}, which has
concrete versions of arguments.

\begin{verbatim}
VecP ◂ Π A : ★ . Π n : Nat . VecC · A n ➔ ★ =
  λ A : ★ . λ n : Nat . λ xsC : VecC · A n .
  ∀ X : Nat ➔ ★ . ∀ P : Π n : Nat . X n ➔ ★ .
  ∀ cN : X zero . ∀ cC : ∀ n : Nat . A ➔ X n ➔ X (suc n) .
  Π pN : P zero cN .
  Π pC : ∀ n : Nat . ∀ xs : X n . 
    Π x : A . P n xs ➔ P (suc n) (cC -n x xs) .
  P n (xsC · X cN cC) .
\end{verbatim}

Notice that the abstract constructors (\texttt{cN} and \texttt{cC})
are implicit arguments, but the abstract parametricity
branches (\texttt{pN} and \texttt{pC}) are explicit arguments
(introduced as non-erased dependent functions via $\Pi$).
Furthermore, the number of explicit (non-erased) arguments in the types of
\texttt{cN} and \texttt{cC} is equal the number of explicit arguments
in the types of \texttt{pN} and \texttt{pC}, respectively. This
coincidence has been arranged so that Church-encoded vectors
(\texttt{VecC}) and their parametricity witnesses (\texttt{VecP})
share the same class of (erased) untyped term inhabitants.

\begin{convention}
We include ``\texttt{P}'' in
the suffix of an identifier to indicate that it relates to
the parametricity theorem of a Church-encoded datatype.
\end{convention}

\paragraph{Constructors}

Now we define ``constructors'' for witnessing parametricity in the
nil case (\texttt{nilPV}) and the cons case (\texttt{consPV}):

\begin{verbatim}
nilPV ◂ ∀ A : ★ . VecP · A zero (nilCV · A) = Λ A .
  Λ X . Λ P . Λ cN . Λ cC . λ pN . λ pC . pN .

consPV ◂ ∀ A : ★ . ∀ n : Nat . ∀ xsC : VecC · A n .
  Π x : A . VecP · A n xsC ➔ 
  VecP · A (suc n) (consCV · A -n x xsC) =
  Λ A . Λ n . Λ xsC . λ x . λ xsP .
  Λ X . Λ P . Λ cN . Λ cC . λ pN . λ pC .
  pC -n -(xsC · X cN cC) x (xsP · X · P -cN -cC pN pC) .
\end{verbatim}

Any additional arguments that would get in the way of the
parametricity witnesses erasing to their corresponding
Church-encodings appear as \textit{implicit} (erased) arguments,
such as \texttt{-(xsC · X cN cC)} in the definition of \texttt{consPV}.
The parametricity witness of the nil (resp. cons) case erases to the
Church-encoding of nil (resp. cons),
just like the erasure of \texttt{nilCV} (resp. \texttt{consCV}).

\begin{lemma}
\erase{\name{nilPV}} is the Church-encoding of nil.
\lablem{nilpv}
\end{lemma}

\begin{proof}
{\small Erase implicit abstractions. \qed}
\end{proof}

\begin{lemma}
\erase{\name{consPV}} is Church-encoded cons.
\lablem{conspv}
\end{lemma}

\begin{proof}
{\small Erase implicit abstractions and applications. \qed}
\end{proof}

\subsection{Reflection Theorem}
\labsec{ind:reflect}

The third (and final) component (\texttt{VecR}) is the reflection
theorem for Church-encoded vectors. It states that eliminating a
vector as a vector, and using its constructors
(\texttt{nilCV} and \texttt{consCV}) for the branches, results in the
vector being eliminated:

\begin{verbatim}
VecR ◂ Π A : ★ . Π n : Nat . VecC · A n ➔ ★ =
  λ A : ★ . λ n : Nat . λ xsC : VecC · A n .
  xsC · (VecC · A) nilCV consCV ≃ xsC .
\end{verbatim}

We cannot derive this using the Church-encoded vector type
(\texttt{VecC}), because it lacks an
induction principle~\cite{geuvers01}. Hence, we
include \texttt{VecR} as a component
of the inductive vector \textit{definition},
which will have an induction principle.

\begin{convention}
We include ``\texttt{R}'' in
the suffix of an identifier to indicate that it relates to
the reflection theorem of a Church-encoded datatype.
\end{convention}

\paragraph{Constructors}

Now we also define ``constructors'' for witnessing reflection in the
nil and cons cases:

\begin{verbatim}
nilRV ◂ ∀ A : ★ . VecR · A zero (nilCV · A) = Λ A . β .

consRV ◂ ∀ A : ★ . ∀ n : Nat .
  ∀ x : A . ∀ xsC : VecC · A n . ∀ q : VecR · A n xsC .
  VecR · A (suc n) (consCV · A -n x xsC)
  = Λ A . Λ n . Λ x . Λ xsC . Λ q . ρ+ q - β .
\end{verbatim}

Reflection for the nil case is proven trivially by \texttt{β}, the
reflexive constructor of equality types. Reflection for the cons case
is proven by first rewriting (using the equality elimination rule
\texttt{ρ}) by the reflection proof for the tail of the vector (\texttt{q}),
after which the proof becomes trivial (\texttt{β}).

\begin{remark}
A full definition of all introduction and elimination rules for our
base theory \cdle can be found in Figure 8 of Stump'17~\cite{stump17b}.
\end{remark}

\subsection{Inductive Type}
\labsec{ind:ind}

Finally, we define the inductive type of vectors (\texttt{Vec}) as the
dependent intersection (using type former $\iota$) of the
Church-encoded vector type
(\texttt{VecC} from \refsec{ind:church}) and its
parametricity theorem
(\texttt{VecP} from \refsec{ind:param}), which is again
intersected with the reflection theorem for
Church-encoded vectors (\texttt{VecR} from \refsec{ind:reflect}):

\begin{verbatim}
Vec ◂ ★ ➔ Nat ➔ ★ = λ A : ★ . λ n : Nat .
  ι xs : (ι xsC : VecC · A n . VecP · A n xsC) . VecR · A n xs.1 .
\end{verbatim}

A dependent intersection ($\iota$-type)
is like a dependent pair ($\Sigma$-type) whose erased components must be equal, and
whose pair constructor erases to its erased left component
($\erase{\pair{t}{t'}} = \erase{t}$).
\texttt{VecC} and \texttt{VecP} share the same class of erased inhabitants, so it
makes sense to intersect them. But, why does it make sense to
intersect these with proofs of equality (\texttt{VecR})?
The answer involves a modified reflexive equality introduction rule,
accepting any term as an additional argument, where the erasure of the
reflexive equality proof becomes the erasure of the term argument
($\erase{\refl{t}} = \erase{t}$).\footnote{A reflexive equality proof
  without a term argument ($\beta$) erases to the identity function by
  default, as in \cdle\cite{stump17b}.
  }

\paragraph{Constructor Helper Function}

Below, we define a helper function to construct a vector from the
intersection of \texttt{VecC} and \texttt{VecP}, and the reflection
theorem (\texttt{VecR}) as an implicit argument (\texttt{➾} is syntax
for non-dependent \texttt{∀}).

\begin{verbatim}
mkVec ◂ ∀ A : ★ . ∀ n : Nat .
  Π xs : (ι xsC : VecC · A n . VecP · A n xsC) .
  VecR · A n xs.1 ➾ Vec · A n =
  Λ A . Λ n . λ xs . Λ q . [ xs , ρ q - β{xs} ] .
\end{verbatim}

The left component of the intersection pair is our
(\texttt{VecC}/\texttt{VecP}) intersection
\texttt{xs}. Although the right component expects the reflection
theorem (\texttt{q}), we cannot return \texttt{q} immediately, because
the intersection pair introduction rule requires the erasure of both
components to be equal. Instead, we rewrite by our reflection proof,
changing the goal from (\texttt{xsC · (VecC · A) nilCV consCV ≃ xsC})
to (\texttt{xsC ≃ xsC}). Then, we use \texttt{β\{xs\}} to construct
a trivial equality that erases to the same term as the left component
of the pair (\texttt{xs}).

\begin{assumption}
  To conserve space, henceforth all proofs assume that implicit
  abstractions and applications have already been erased.
\labasm{eraseimp}
\end{assumption}

Below, we can see that \texttt{mkVec} is our first example of an
\textit{identity coercion} (a function erasing to the
identity). Additionally, the proof demonstrates how the intersection
pair components erase to the same term ($xs$), making it a well-typed
introduction of an intersection pair:

\begin{lemma}
\erase{\name{mkVec}} is the identity function:
\lablem{mkvec}
\end{lemma}

\begin{proof}
{\small
\abovedisplayskip=-\baselineskip
\begin{align*}
  &=_\delta \fun{xs} \erase{
    \pair{xs}{ \Erase{ \rewrite{q}{\refl{xs}} } }
  }
  && \by{Erase rewrite.}
  \\
  &=~ \fun{xs} \erase{
    \pair{xs}{ \Erase{\refl{xs}} }
  }
  && \by{Erase reflexive equality.}
  \\
  &=~ \fun{xs} \erase{
    \pair{xs}{xs}
  }
  && \by{Erase pair.}
  \\
  &=~ \fun{xs}{xs}
  \eqed
\end{align*}}
\end{proof}

\begin{notation}
  We use large pipes (within small pipes) to focus on the erasure of
  subterms, rather than erasing according to
  a depth-first strategy. For example, \erase{f \earg x \earg \Erase{y}}
  denotes erasing the subterm $y$ first.
\end{notation}

\paragraph{Constructors}

Finally, it is straightforward to define constructors of our inductive
vector type (\texttt{Vec}) from the helper \texttt{mkVec} and the 3
constructor components we defined previously.

\begin{verbatim}
nilV ◂ ∀ A : ★ . Vec · A zero = Λ A .
  mkVec · A -zero [ nilCV · A , nilPV · A ] -(nilRV · A) .

consV ◂ ∀ A : ★ . ∀ n : Nat. A ➔ Vec · A n ➔ Vec · A (suc n) =
  Λ A . Λ n . λ x . λ xs . mkVec · A -(suc n)
  [ consCV · A -n x xs.1.1 , consPV · A -n -xs.1.1 x xs.1.2 ]
  -(consRV · A -n -x -xs.1.1 -xs.2) .
\end{verbatim}

Below, we verify that the inductive constructors erase to their untyped
Church-encoded equivalents:

\begin{theorem}
\erase{\name{nilV}} is the Church-encoding of nil:
\labthm{nilv}
\end{theorem}

\begin{proof}
{\small
\abovedisplayskip=-\baselineskip
\begin{align*}
  &=_\delta \erase{
    \name{mkVec} \earg \pair{ \Erase{\name{nilCV}} }{ \name{nilPV} }
  }
  && \by{By \reflem{nilcv}.}
  \\
  &=~ \erase{
    \name{mkVec} \earg \pair{ \fun{c_n, c_c} c_n }{ \Erase{\name{nilPV}} }
  }
  && \by{By \reflem{nilpv}.}
  \\
  &=_\alpha \erase{
    \name{mkVec} \earg \Erase{\pair{ \fun{c_n, c_c} c_n }{ \fun{p_n, p_c} p_n }}
  }
  && \by{Erase pair.}
  \\
  &=~ \erase{
    \name{mkVec} \earg (\fun{c_n, c_c} c_n)
  }
  && \by{By \reflem{mkvec}.}
  \\
  &=_\beta \fun{c_n, c_c} c_n
  \eqed
\end{align*}}
\end{proof}

In the proof of \refthm{nilv} above, we can see that the intersection
pair passed to \texttt{mkVec} is type correct, as both of its
components erase to the same ($\alpha$-equivalent) term.

\begin{theorem}
\erase{\name{consV}} is the Church-encoding of cons:
\labthm{consv}
\end{theorem}

\begin{proof}
{\small
  Same as the proof of \refthm{nilv}, but erasing
  \name{consCV} (instead of \name{nilCV})
  by \reflem{conscv} in the first step, and
  \name{consPV} (instead of \name{nilPV})
  by \reflem{conspv} in the second step. \qed
}
\end{proof}

\paragraph{Eliminator}

The whole point of defining the inductive vector type (\texttt{Vec}),
as opposed to the Church-encoded vector type (\texttt{VecC}), is so
we can define its eliminator (i.e. its induction principle in type theory):

\begin{verbatim}
elimVec ◂ ∀ A : ★ . ∀ n : Nat . Π xs : Vec · A n .
  ∀ P : Π n : Nat . Vec · A n ➔ ★ .
  Π pN : P zero (nilV · A) .
  Π pC : ∀ n : Nat . ∀ xs : Vec · A n . Π x : A . 
    P n xs ➔ P (suc n) (consV · A -n x xs) .
  P n xs
  = Λ A . Λ n . λ xs . Λ P . ρ ς xs.2 -
  (xs.1.2 · (Vec · A) · P -(nilV · A) -(consV · A)) .
\end{verbatim}

We apply the parametricity theorem (via projection
\texttt{xs.1.2}) to the motive \texttt{P} and the concrete vector
constructors \texttt{nilV} and \texttt{consV}, instantiating the
abstract constructor arguments of \texttt{VecP}. The result has the
following type:

\begin{verbatim}
P n (xs.1.1 · (VecC · A) (nilV · A) (consV · A))
\end{verbatim}

Note that the second argument to \texttt{P} is exactly one of the
sides of our reflection theorem (\texttt{VecR}), so we rewrite
(using \texttt{ρ}) by the
reflection proof (via its projection \texttt{xs.2}) to arrive at our
goal (\texttt{P n xs}). \footnote{Because \texttt{ρ} rewrites the left
  side of an equality to the right side, we get the symmetric version
  of our reflection proof \texttt{xs.2} by applying our symmetry
  operator \texttt{ς}. The operator \texttt{ς} changes any equation
  \texttt{t1 ≃ t2} to \texttt{t2 ≃ t1}.
  }

Because the projections and rewrites are erased, the eliminator
\texttt{elimVec} is actually a dependent identity coercion
(from \texttt{Vec} to the rest of the eliminator type, starting
with \texttt{∀ P} and ending with \texttt{P n xs}):

\begin{lemma}
\erase{\name{elimVec}} is the identity function:
\lablem{elimvec}
\end{lemma}

\begin{proof}
{\small
\abovedisplayskip=-\baselineskip
\begin{align*}
  &=_\delta \fun{xs} \erase{
    \rewrite{
      \varsigma \earg \Erase{xs.2}
    }{
      \Erase{xs.1.2}
    }
  }
  && \by{Erase projections.}
  \\
  &=~ \fun{xs} \erase{
    \rewrite{
      \varsigma \earg xs
    }{
      xs
    }
  }
  && \by{Erase rewrite.}
  \\
  &=~ \fun{xs}{xs}
  \eqed
\end{align*}}
\end{proof}

\section{Reusing Vector Definitions}
\labsec{vecreuse}

In this section we demonstrate reusing vector programs and proofs to
define list-versions of the programs and proofs. Through the use of
identity coercions, our program reuse does not introduce runtime overhead,
and our proof reuse does not introduce equational reasoning overhead.

\subsection{Identity Coercion from Vec to List}
\labsec{vecreuse:v2l}

We extend Barras and Bernardo's~\cite{barras:implicit} identity
coercion from Church-encoded vectors to lists (\texttt{v2lC}),
to an identity coercion between inductive versions
of the types (\texttt{v2l}), i.e. those supporting induction
principles.
Identity coercions for inductive types are defined
using the same 3 components as inductive constructors (a
Church-encoding component, like in \refsec{ind:church},
a parametricity theorem component, like in \refsec{ind:param}, and a
reflection theorem component, like in \refsec{ind:reflect}).

\paragraph{Church-Encoding}

\begin{figure}[t]
\centering
\begin{verbatim}
v2lC' ◂ ∀ A : ★ . ∀ n : Nat . VecC · A n ➔ ListC · A
  = Λ A . Λ n . λ xs .
  xs · (λ _ : Nat . ListC · A) (nilCL · A) (Λ _ . consCL · A) .
\end{verbatim}
\caption{Non-identity coercion from vectors to lists.}
\labfig{nonident}
\end{figure}

A standard way of translating a Church-encoded vector to a
Church-encoded list is to eliminate the vector at the \textit{concrete} list
type, as show in \reffig{nonident}. Alternatively, we can
``go underneath'' the list codomain, and eliminate the vector using
the \textit{abstract} list return type (\texttt{X}) and
\textit{abstract} list constructors
(\texttt{cN} for nil and \texttt{cC} for cons). This alternative way,
which is possible when the domain is a \textit{subtype} of the
codomain, appears below.

\begin{verbatim}
v2lC ◂ ∀ A : ★ . ∀ n : Nat . Vec · A n ➔ ListC · A
  = Λ A . Λ n . λ xs . Λ X . λ cN . λ cC .
  xs.1.1 · (λ _ : Nat . X) cN (Λ _ . cC) .
\end{verbatim}

One minor difference, compared to \texttt{v2lC'} in \reffig{nonident},
is that \texttt{v2lC} takes an inductive (rather than Church) vector as its
argument. Hence, we access the Church-encoded vector via the
projection \texttt{xs.1.1}.
This difference only becomes necessary in
\refsec{vecreuse:l2v}, where it allows us to define a
\textit{dependent} identity coercion.

Barras and Bernardo point out that after erasure, the alternative
abstract elimination $\eta$-contracts to the identity
function, and for this reason we call it an
``identity coercion'':

\begin{lemma}
\erase{\name{v2lC}} is the identity function:
\lablem{v2lc}
\end{lemma}

\begin{proof}
{\small
\abovedisplayskip=-\baselineskip
\begin{align*}
  &=_\delta \fun{xs} \fun{c_n} \fun{c_c} \erase{
    xs.1.1 \earg c_n \earg c_c
  }
  && \by{Erase projection.}
  \\
  &=~ \fun{xs} (\fun{c_n} \fun{c_c}
    xs \earg c_n \earg c_c)
  && \by{Contract.}
  \\
  &=_\eta \fun{xs}{xs}
  \eqed
\end{align*}}
\end{proof}

\paragraph{Parametricity Theorem}

Second, we translate the vector parametricity theorem to the list
parametricity theorem, this time projecting out the vector
parametricity theorem (via \texttt{xs.1.2}). In addition to the abstract
arguments that \texttt{v2lC} receives from its codomain,
\texttt{v2lP} also receives an abstract motive (\texttt{P}) and
abstract parametricity theorem branches
(\texttt{pN} and \texttt{pC}).

\begin{verbatim}
v2lP ◂ ∀ A : ★ . ∀ n : Nat .
  Π xs : Vec · A n . ListP · A (v2lC · A -n xs)
  = Λ A . Λ n . λ xs . Λ X . Λ P . Λ cN . Λ cC . λ pN . λ pC .
  xs.1.2 · (λ _ : Nat . X) · (λ _ : Nat . P)
    -cN -(Λ _ . cC) pN (Λ _ . pC) .
\end{verbatim}

Because the abstract motive (\texttt{P}) is an implicit argument,
and the abstract Church constructors (\texttt{cN} and \texttt{cC}) are
also implicit arguments, they get erased,
thus \texttt{v2lP} is also an identity coercion
(albeit between parametricity theorems):

\begin{lemma}
\erase{\name{v2lP}} is the identity function:
\lablem{v2lp}
\end{lemma}

\begin{proof}
{\small
\abovedisplayskip=-\baselineskip
\begin{align*}
  &=_\delta \fun{xs} \fun{p_n} \fun{p_c} \erase{
    xs.1.2 \earg p_n \earg p_c
  }
  && \by{Erase projection.}
  \\
  &=~ \fun{xs} (\fun{p_n} \fun{p_c}
    xs \earg p_n \earg p_c)
  && \by{Contract.}
  \\
  &=_\eta \fun{xs}{xs}
  \eqed
\end{align*}}
\end{proof}

\paragraph{Reflection Theorem}

Third, we reuse the Church vector reflection theorem
(projection \texttt{xs.2}) to prove the vector  reflection theorem.

\begin{verbatim}
v2lR ◂ ∀ A : ★ . ∀ n : Nat .
  Π xs : Vec · A n . ListR · A (v2lC · A -n xs)
  = Λ A . Λ n . λ xs . xs.2 .
\end{verbatim}

\begin{convention}
We suffix an identifier with ``\texttt{L}''
to indicate that it relates to lists.
\end{convention}

A proof of vector reflection
(\texttt{xsC · (VecC · A) nilCV consCV ≃ xsC}) can be used as a proof
of list reflection
(\texttt{xsC · (ListC · A) nilCL consCL ≃ xsC}),
because their types are definitionally equal (where definitional
equality is defined on erased terms). This works because the type
applications (\texttt{VecC · A} and \texttt{ListC · A}) are erased,
\texttt{nilCV} and \texttt{nilCL} both erase to
the untyped Church-encoding of nil
(by \reflem{nilcv} for vectors, and similarly for lists), and
\texttt{consCV} and \texttt{consCL} both erase to
the untyped Church-encoding of cons
(by \reflem{conscv} for vectors, and similarly for lists).

\paragraph{Identity Coercion}

Finally, we put together our 3 components
(\texttt{v2lC}, \texttt{v2lP}, and \texttt{v2lR}) to translate
inductive vectors to inductive lists, using the
\texttt{mkList} helper constructor. This is analogous to defining the
vector constructors in terms of their 3 components and \texttt{mkVec}
in \refsec{ind:ind}.

\begin{verbatim}
v2l ◂ ∀ A : ★ . ∀ n : Nat . Vec · A n ➔ List · A
  = Λ A . Λ n . λ xs . mkList · A
  [ v2lC · A -n xs , v2lP · A -n xs ] -(v2lR · A -n xs) .
\end{verbatim}

We have successfully generalized Barras and Bernardo's non-dependent
identity coercion between \textit{Church-encoded} vectors and lists, to a
non-dependent identity coercion between \textit{inductive} vectors and lists:

\begin{theorem}
\erase{\name{v2l}} is the identity function:
\labthm{v2l}
\end{theorem}

\begin{proof}
{\small
\abovedisplayskip=-\baselineskip
\begin{align*}
  &=_\delta \fun{xs} \erase{
    \name{mkList} \earg
    \pair{
      \Erase{\name{v2lC} \earg xs}
    }{
      \name{v2lP} \earg xs
    }
  }
  && \by{By \reflem{v2lc}.}
  \\
  &=_\beta \fun{xs} \erase{
    \name{mkList} \earg
    \pair{
      xs
    }{
      \Erase{\name{v2lP} \earg xs}
    }
  }
  && \by{By \reflem{v2lp}.}
  \\
  &=_\beta \fun{xs} \erase{
    \name{mkList} \earg
    \Erase{\pair{xs}{xs}}
  }
  && \by{Erase pair.}
  \\
  &=~ \fun{xs} \erase{
    \name{mkList} \earg xs
  }
  && \by{By \reflem{mkvec} (for lists).}
  \\
  &=_\beta \fun{xs}{xs}
  \eqed
\end{align*}}
\end{proof}

Type checking requires that both components of
the pair (introducing an intersection type), in the definition of
\texttt{v2l}, are definitionally equal. The
third step in the proof of \refthm{v2l} demonstrates that this
requirement is satisfied, after erasing the left and right components
in the first and second steps, respectively.

\begin{remark}
Although the coercion of the reflection theorem (\texttt{v2lR})
happens to erase to the identity function, we never emphasize the erasure of
reflection theorem proofs. This is because they appear in erased
argument positions in the definitions of identity coercions (e.g.
the erased argument \texttt{(v2lR · A -n xs)}
of \texttt{mkList}, in the definition of \texttt{v2l}).
\labrem{reflect}
\end{remark}

\begin{figure}[t]
\centering
\begin{verbatim}
l2vC ◂ ∀ A : ★ . Π xs : List · A . VecC · A (length · A xs)
  = Λ A . λ xs . Λ X . λ cN . λ cC . elimList · A xs ·
  (λ xs : List · A . X (length · A xs))
  cN (Λ xs . cC -(length · A xs)) .

l2vP ◂ ∀ A : ★ . Π xs : List · A . VecP · A (length · A xs) (l2vC · A xs)
  = Λ A . λ xs . Λ X . Λ P . Λ cN . Λ cC . λ pN . λ pC . elimList · A 
  xs · (λ xs : List · A . P (length · A xs) (l2vC · A xs · X cN cC))
  pN (Λ xs . pC -(length · A xs) -(l2vC · A xs · X cN cC)) .

l2vR ◂ ∀ A : ★ . Π xs : List · A . VecR · A (length · A xs) (l2vC · A xs)
  = Λ A . λ xs . xs.2 .

l2v ◂ ∀ A : ★ . Π xs : List · A . Vec · A (length · A xs)
  = Λ A . λ xs . mkVec · A -(length · A xs)
  [ l2vC · A xs , l2vP · A xs ] -(l2vR · A xs) .
\end{verbatim}
\caption{Identity coercion from lists to vectors.}
\labfig{l2v}
\end{figure}

\subsection{Identity Coercion from List to Vec}
\labsec{vecreuse:l2v}

Barras and Bernardo's \textit{non-dependent} identity coercion takes
Church-encoded vectors to lists, which we have extended
(in \refsec{vecreuse:v2l}) to take inductive vectors to lists.
Because we are using inductive types, we can now define the
\textit{dependent} identity coercion from inductive lists to vectors
(\texttt{l2v}). Church-encoded types cannot be used to define
\texttt{l2v}, as the resulting vector length depends on the input
vector in the type of \texttt{l2v}
(i.e. \texttt{Π xs : List · A . Vec · A (length · A xs)}, as in
\reffig{l2v}). Although we could express the type using Church-encoded
data, we could not inhabit it, as reducing \texttt{length} in the
codomain (when defining the nil and cons branches of the coercion)
requires an induction principle (i.e. \texttt{elimList}).

\reffig{l2v} contains the definition of \texttt{l2v}, which follows
the same 3-component structure used to define \texttt{v2l} in
\refsec{vecreuse:l2v}. The primary difference is that the
Church (\texttt{l2vC}) and parametricity (\texttt{l2vP}) components
are defined by \textit{induction}, using \texttt{elimList}.
It is crucial that the domains of \texttt{l2vC} and \texttt{l2vP} are
\textit{inductive} lists, because \texttt{l2vC} and \texttt{l2vP} need
to be defined by induction.

\begin{lemma}
\erase{\name{l2vC}} is the identity function:
\lablem{l2vc}
\end{lemma}

\begin{proof}
{\small
\abovedisplayskip=-\baselineskip
\begin{align*}
  &=_\delta \fun{xs} \fun{c_n} \fun{c_c} \erase{
    \name{elimList} \earg xs \earg c_n \earg c_c
  }
  && \by{By \reflem{elimvec} (for lists).}
  \\
  &=_\beta \fun{xs} (\fun{c_n} \fun{c_c}
    xs \earg c_n \earg c_c)
  && \by{Contract.}
  \\
  &=_\eta \fun{xs}{xs}
  \eqed
\end{align*}}
\end{proof}

It is not enough that we can define a coercion from lists to vectors,
we also want \texttt{l2v} to be an \textit{identity} coercion. This is
established by \refthm{l2v}, which relies on \texttt{l2vC} being an
identity coercion (\reflem{l2vc}), and \texttt{l2vP} being an identity
coercion (\reflem{l2vp}).
The proof that \texttt{l2vC} is an identity coercion is similar to
Barras and Bernardo's argument about \texttt{v2lC}
(\reflem{v2lc}), relying on erasure and $\eta$-contraction. The main
difference is that we also rely on the fact that our induction
principle (\texttt{elimList}) is \textit{also} an identity coercion
(\reflem{elimvec}, but for the list datatype).

\begin{lemma}
\erase{\name{l2vP}} is the identity function:
\lablem{l2vp}
\end{lemma}

\begin{proof}
{\small
  Same as the proof of \reflem{l2vc},
  $\alpha$-renaming $c_n$ and $c_c$ to
  $p_n$ and $p_c$. \qed
}
\end{proof}

\begin{theorem}
\erase{\name{l2v}} is the identity function:
\labthm{l2v}
\end{theorem}

\begin{proof}
{\small
  Same as the proof of \refthm{v2l}, but erasing
  \name{l2vC} (instead of \name{v2lC})
  by \reflem{l2vc} in the first step,
  \name{l2vP} (instead of \name{v2lP})
  by \reflem{l2vp} in the second step, and
  \name{mkVec} (instead of \name{mkList})
  by \reflem{mkvec} in the final step. \qed
}
\end{proof}

\subsection{Program Reuse}
\labsec{vecreuse:progreuse}

We achieve program reuse by defining list append
(\texttt{appendL}) in terms of vector append (\texttt{appendV}), in the
standard way by applying \texttt{appendV} to the result of coercing
both arguments to vectors (using \texttt{l2v}), and coercing the
result of vector append to a list (using \texttt{v2l}).

\begin{verbatim}
appendL ◂ ∀ A : ★ . List · A ➔ List · A ➔ List · A
  = Λ A . λ xs . λ ys . v2l · A
  -(add (length · A xs) (length · A ys))
  (appendV · A -(length · A xs) (l2v · A xs)
    -(length · A ys) (l2v · A ys)) .
\end{verbatim}

The important property is that program reuse (i.e. the definition of
\texttt{appendL} in terms of \texttt{appendV}) incurs no runtime
penalty. We prove this below, showing that the erasure of our derived
list append is equal to the erasure of vector append, which relies on
\texttt{v2l} and \texttt{l2v} being identity coercions:

\begin{theorem}
\erase{\name{appendL}} is \erase{\name{appendV}}:
\labthm{appendl}
\end{theorem}

\begin{proof}
{\small
\abovedisplayskip=-\baselineskip
\begin{align*}
  &=_\delta \fun{xs} \fun{ys} \erase{
    \name{v2l} \earg (
    \name{appendV} \earg
    \Erase{\name{l2v} \earg xs} \earg \Erase{\name{l2v} \earg ys}
    )
  }
  && \by{By \refthm{l2v}.}
  \\
  &=_\beta \fun{xs} \fun{ys} \erase{
    \name{v2l} \earg (
    \name{appendV} \earg xs \earg ys
    )
  }
  && \by{By \refthm{v2l}.}
  \\
  &=_\beta \fun{xs} \fun{ys}
    \erase{\name{appendV}} \earg xs \earg ys
  && \by{Contraction.}
  \\
  &=_\eta \erase{\name{appendV}}
  \eqed
\end{align*}}
\end{proof}

\subsection{Proof Reuse}
\labsec{vecreuse:proofreuse}

Proof reuse, proving that list append is associative
(\texttt{appendAssocL}) in terms of a proof that
vector append is associative (\texttt{appendAssocV}),
is even easier than program reuse. We derive
\texttt{appendAssocL} by applying
\texttt{appendAssocV} to the result of coercing each argument from a
list to a vector (using \texttt{l2v}). We do not need to coerce in the
other direction (using \texttt{v2l}), because our result is already an
equality type that erases to our goal.\footnote{Our proof reuse does
  not require congruence (\texttt{cong}), as used in the introduction
  \refsec{intro}, because converting two of our equality types
  (\texttt{≃}) only requires their erased terms to be equal, not their
  types.
  }

\begin{verbatim}
appendAssocL ◂ ∀ A : ★ .
  Π xs : List · A . Π ys : List · A . Π zs : List · A .
  appendL (appendL xs ys) zs ≃ appendL xs (appendL ys zs)
  = Λ A . λ xs . λ ys . λ zs . appendAssocV · A
    -(length · A xs) (l2v · A xs)
    -(length · A ys) (l2v · A ys)
    -(length · A zs) (l2v · A zs) .
\end{verbatim}

The result of reusing \texttt{appendAssocV} has the following type:
\begin{verbatim}
appendV (appendV (l2v xs) (l2v ys)) (l2v zs) ≃ 
  appendV (l2v xs) (appendV (l2v ys) (l2v zs))
\end{verbatim}

After erasure, this $\beta$-reduces to our goal because
\texttt{appendV} erases to \texttt{appendL} by
\refthm{appendl}, and \texttt{l2v} erases to (\texttt{λ x . x})
by \refthm{l2v}.
Without identity coercions in the derived program \texttt{appendL} and
derived proof \texttt{appendAssocL}, proof reuse would require equational
reasoning by appealing to the identity laws established by an
isomorphism between lists and vectors (as mentioned in the
introduction \refsec{intro}).

\section{Reusing List Definitions}
\labsec{listreuse}

In this section we demonstrate reuse in the other direction
(compared to \refsec{vecreuse}), reusing list programs and proofs to
define vector-versions of the programs and proofs.
This direction of reuse takes more effort, because we may want to write
functions over vectors with \textit{index constraints} in terms of functions
over lists without the constraints. Hence, we are required to prove
that the list-based reused definition implies the constraints we
explicitly state in the vector-based derived definition.

\subsection{Vectors as Length-Constrained Lists}
\labsec{listreuse:vecl}

The \texttt{v2l} function is ``lossy'' in the sense that the input
vector length does not appear in the list codomain.
In \refsec{listreuse:v2u}, we create a version of \texttt{v2l}
(named \texttt{v2u}) that ``remembers'' the index information, by
taking a vector to a list and a constraint on its length. Below, we
derive a new type (which will be the codomain of \texttt{v2u}),
named \texttt{VecL}, as the intersection of a list
and its length constraint.

\begin{verbatim}
VecL ◂ ★ ➔ Nat ➔ ★ = λ A : ★ . λ n : Nat .
  ι xs : List · A . n ≃ length · A xs .
\end{verbatim}

Because we use an intersection type, the erasure of \texttt{VecL} will
always be its left (\texttt{List}) component, so the additional
constraint does not get in the way of definitional equality
checking.

Below, we define the constructor helper function \texttt{mkVecL},
taking a list and an erased constraint to a \texttt{VecL}.
We rewrite by the proof of the constraint (\texttt{q}) in the right
component, so that we may return \texttt{β\{xs\}} as the right
component, allowing the erasure of the left and right sides to both be
\texttt{xs}.

\begin{verbatim}
mkVecL ◂ ∀ A : ★ . ∀ n : Nat . Π xs : List · A .
  (length · A xs ≃ n) ➾ VecL · A n =
  Λ A . Λ n . λ xs . Λ q . [ xs , ρ q - β{xs} ] .
\end{verbatim}

Just like \texttt{mkVec} in \refsec{ind:ind}, \texttt{mkVecL} is also
an identity coercion:

\begin{lemma}
\erase{\name{mkVecL}} is the identity function:
\lablem{mkvecl}
\end{lemma}

\begin{proof}
{\small
Same as the proof of \reflem{mkvec}. \qed
}
\end{proof}

\subsection{Identity Coercion from Vec to VecL}
\labsec{listreuse:v2u}

\begin{figure}[t]
\centering
\begin{verbatim}
lengthPres ◂ ∀ A : ★ . ∀ n : Nat . Π xs : Vec · A n .
  n ≃ length · A (v2l · A -n xs)
\end{verbatim}
\caption{Length is preserved by coercion.}
\labfig{lenpres}
\end{figure}

Now we define \texttt{v2u}, taking a vector to a list and the constraint
that the length of the list is equal to the index of the vector
(by using \texttt{VecL} as the codomain of \texttt{v2u}). The function
\texttt{v2u} uses \texttt{mkVecL} to construct a \texttt{VecL} from a
vector by coercing to a list (via \texttt{v2l}), and proving the
constraint that \texttt{v2l} preserves the vector index length
w.r.t. the output list length
(via \texttt{lengthPres} in \reffig{lenpres}):

\begin{verbatim}
v2u ◂ ∀ A : ★ . ∀ n : Nat . Vec · A n ➔ VecL · A n
  = Λ A . Λ n . λ xs . mkVecL · A -n
  (v2l · A -n xs) -(ς (lengthPres · A -n xs)) .
\end{verbatim}

\begin{convention}
We include ``u'' in an identifier
to indicate that it relates to length-constrained lists.
\end{convention}

The function \texttt{v2u} is also an identity coercion, as it is
defined in terms of other identity coercions (\texttt{mkListL} and
\texttt{v2l}):

\begin{theorem}
\erase{\name{v2u}} is the identity function:
\labthm{v2u}
\end{theorem}

\begin{proof}
{\small
\abovedisplayskip=-\baselineskip
\begin{align*}
  &=_\delta \fun{xs} \erase{
    \name{mkListL} \earg
    \Erase{\name{v2l} \earg xs}
  }
  && \by{By \refthm{v2l}.}
  \\
  &=_\beta \fun{xs} \erase{
    \name{mkListL} \earg xs
  }
  && \by{By \reflem{mkvecl}.}
  \\
  &=_\beta \fun{xs}{xs}
  \eqed
\end{align*}}
\end{proof}

\subsection{Program Reuse}

We achieve program reuse by defining vector append
(\texttt{appendV}) in terms of list append (\texttt{appendL}). We must
coerce the output of \texttt{appendL} to a vector by \texttt{l2v},
and the arguments of \texttt{appendL} to lists by the first
projection of \texttt{v2u}. However, we must also perform rewrites to
ensure that \texttt{appendV} produces a vector whose length is the sum
of both input vectors (\texttt{add n m}):

\begin{verbatim}
appendV ◂ ∀ A : ★ .
  ∀ n : Nat . Vec · A n ➔
  ∀ m : Nat . Vec · A m ➔
  Vec · A (add n m)
  = Λ A . Λ n . λ xs . Λ m . λ ys .
  ρ (v2u · A -n xs).2 - 
  ρ (v2u · A -m ys).2 -
  ρ (lengthDistAppend · A (v2u · A -n xs).1 (v2u · A -m ys).1) -
  (l2v · A (appendL · A (v2u · A -n xs).1 (v2u · A -m ys).1)) .
\end{verbatim}

The result of reusing \texttt{appendL} has the following type:

\begin{verbatim}
Vec · A (length (appendL (v2u xs).1 (v2u  ys).1))
\end{verbatim}

\begin{figure}[t]
\centering
\begin{verbatim}
lengthDistAppend ◂ ∀ A : ★ . Π xs : List · A . Π ys : List · A .
  add (length · A  xs) (length · A ys) ≃ length (appendL · A xs ys)
\end{verbatim}
\caption{Length distributes through append.}
\labfig{lendistappend}
\end{figure}

We rewrite by the property
(\texttt{lengthDistAppend} in \reffig{lendistappend})
that length distributes through list
append via addition:

\begin{verbatim}
Vec · A (add (length (v2u xs).1) (length  (v2u ys).1))
\end{verbatim}

Rewriting by the length-constraints of both coerced lists, via
the second projection of \texttt{v2u} for \texttt{xs} and
\texttt{ys}, results in our goal type.

Reusing the program \texttt{appendL} to define \texttt{appendV} incurs
no runtime penalty:

\begin{theorem}
\erase{\name{appendV}} is \erase{\name{appendL}}:
\labthm{appendv}
\end{theorem}

\begin{proof}
{\small
  Erase rewrites and projections, then same as the proof of
  \refthm{appendl}, exchanging \texttt{appendL} for \texttt{appendV},
  and \texttt{v2u} (erased by \refthm{v2u}) for \texttt{v2l}. \qed
}
\end{proof}

\subsection{Proof Reuse}

We achieve proof reuse by proving that vector append is associative
(\texttt{appendAssocV}) in terms of a proof that
list append is associative (\texttt{appendAssocL}).
Once again, this is easier than program reuse, as we must only
coerce the arguments to lists (using \texttt{v2l}), but must not
coerce the result (using \texttt{l2v}) because it is already an
equality type:

\begin{verbatim}
appendAssocV ◂ ∀ A : ★ .
  ∀ n : Nat . Π xs : Vec · A n .
  ∀ m : Nat . Π ys : Vec · A m .
  ∀ o : Nat . Π zs : Vec · A o .
  appendV (appendV xs ys) zs ≃ appendV xs (appendV ys zs)
  = Λ A . Λ n . λ xs . Λ m . λ ys . Λ o . λ zs .
  appendAssocL · A (v2l · A -n xs) (v2l · A -m ys) (v2l · A -o zs) .
\end{verbatim}

The result of reusing \texttt{appendAssocL} has the following type:
\begin{verbatim}
appendL (appendL (l2v xs) (l2v ys)) (l2v zs) ≃ 
  appendL (l2v xs) (appendL (l2v ys) (l2v zs))
\end{verbatim}

After erasure, this $\beta$-reduces to our goal because
\texttt{appendL} erases to \texttt{appendV} by
\refthm{appendv}, and \texttt{l2v} erases to (\texttt{λ x . x})
by \refthm{l2v}. Again, proof reuse is zero-cost as no equational
reasoning needs to be performed.

\begin{remark}
Note that in the definition of both \texttt{appendV} and
\texttt{appendAssocV}, we could exchange \texttt{(v2l xs)} for
\texttt{(v2u xs).1} and lemma \texttt{(lengthPres xs)} for
\texttt{(v2u xs).2}, and vice versa, because the latter term in both pairs
erases to the former term.
\end{remark}

\section{Reusing Nested List Definitions}
\labsec{nestreuse}

In this section we demonstrate reuse for nested datatypes,
reusing programs and proofs over lists of lists
(\texttt{List · (List · A)}) to define programs and proofs
over vectors of vectors (\texttt{Vec · (Vec · A n) m}).
Like in \refsec{listreuse}, such reuse requires proving properties
about the lists to satisfy the vector length requirements of
the derived definitions.

\subsection{List Map}
\labsec{nestreuse:mapl}

To reuse a list of lists as a vector of vectors, we must be able to
coerce the inner lists in addition to the outer list. This can be
achieved by mapping \texttt{v2l} over the inner lists. However, to
ensure that reused definitions are identity coercions, it is crucial
that we define list map (\texttt{mapL}) in Barras and Bernardo's style
of eliminating the input list at the abstract type, and using the
abstract constructors, of the output list
(as in \refsec{ind:church}). We define \texttt{mapL} in
abstract-elimination style for inductive lists in terms of our
familiar 3 components (\texttt{mapCL}, \texttt{mapPL}, and
\texttt{mapRL}).

\paragraph{Church-Encoding}

The Church-component of map eliminates the Church-encoded input list
(via projection \texttt{xs.1.1}) at abstract type \texttt{X}, using
abstract constructors \texttt{cN} and \texttt{cC}. We define the head
of our mapped list to be \texttt{(f x)} in
the abstract cons case (\texttt{cC}):

\begin{verbatim}
mapCL ◂ ∀ A : ★ . ∀ B : ★ .  Π f : A ➔ B . List · A ➔ ListC · B
  = Λ A . Λ B . λ f . λ xs . Λ X . λ cN . λ cC . 
  xs.1.1 · X cN (λ x . cC (f x)) .
\end{verbatim}

Even though we have defined \texttt{mapCL} in abstract elimination
style, it is \textit{not} an identity coercion (we know nothing about
the input function \texttt{f}, which may change the elements of the
list). However, if we partially apply \texttt{mapCL} to the identity
function, then the result \textit{is} an identity coercion:

\begin{lemma}
$\erase{\name{mapCL}} \earg (\fun{x}{x})$ is the identity function:
\lablem{mapcl}
\end{lemma}

\begin{proof}
{\small
\begin{align*}
  &=_\delta (\fun{f} \fun{xs} \fun{c_n} \fun{c_c} \erase{
    xs.1.1 \earg c_n \earg (\fun{x} c_c \earg (f \earg x))
  }) \earg (\fun{x}{x})
  && \by{Erase projection.}
  \\
  &=~ (\fun{f} \fun{xs} \fun{c_n} \fun{c_c}
    \earg c_n \earg (\fun{x} c_c \earg (f \earg x))
  ) \earg (\fun{x}{x})
  && \by{Reduce.}
  \\
  &=_\beta \fun{xs} \fun{c_n} \fun{c_c}
    \earg c_n \earg (\fun{x} c_c \earg x)
  && \by{Contract.}
  \\
  &=_\eta \fun{xs} (\fun{c_n} \fun{c_c}
    xs \earg c_n \earg c_c)
  && \by{Contract.}
  \\
  &=_\eta \fun{xs}{xs}
  \eqed
\end{align*}}
\end{proof}

\paragraph{Parametricity Theorem}

The parametricity-component of map eliminates the parametricity theorem input
(via projection \texttt{xs.1.2}) at abstract type \texttt{X} and
abstract motive \texttt{P}, and using abstract constructors \texttt{cN} and
\texttt{cC}, and abstract parametricity branches \texttt{pN} and
\texttt{pC}. This time we use \texttt{(f x)} for the head position of
the cons branch of the parametricity theorem (\texttt{pC}):

\begin{verbatim}
mapPL ◂ ∀ A : ★ . ∀ B : ★ .
  Π f : A ➔ B . Π xs : List · A .
  ListP · B (mapCL · A · B f xs)
  = Λ A . Λ B . λ f . λ xs .
  Λ X . Λ P . Λ cN . Λ cC . λ pN . λ pC .
  xs.1.2 · X · P -cN -(λ x . cC (f x))
    pN (Λ xsC . λ x . pC -xsC (f x)) .
\end{verbatim}

The partial application of the parametricity-component of map to the
identity function is likewise an identity coercion:

\begin{lemma}
$\erase{\name{mapPL}} \earg (\fun{x}{x})$ is the identity function:
\lablem{mappl}
\end{lemma}

\begin{proof}
{\small
  Same as the proof of \reflem{mapcl}, exchanging
  $xs.1.2$ for $xs.1.1$, and $\alpha$-renaming
  $p_n$ and $p_c$ to $c_n$ and $c_c$. \qed
}
\end{proof}

\paragraph{Reflection Theorem}

The reflection theorem component of the map is defined by a simple
induction (using \texttt{elimList}) on the input list, and reusing the
reflection theorem proof of the input:

\begin{verbatim}
mapRL ◂ ∀ A : ★ . ∀ B : ★ .  Π f : A ➔ B . Π xs : List · A .
  ListR · B (mapCL · A · B f xs)
  = Λ A . Λ B . λ f . λ xs . elimList · A xs ·
  (λ xs : List · A . ListR · B (mapCL · A · B f xs))
  β
  (Λ xs . λ x . λ ih . ρ+ ih - β) . 
\end{verbatim}

\begin{remark}
Neither \texttt{mapRL}, nor its partial application to the identity
function, results in an identity coercion. This is not important, as
explained in \refrem{reflect}, because we will only use it in an
erased position in the definition of \texttt{mapL} (as an implicit
argument to \texttt{mkList}).
\end{remark}

\paragraph{List Map}

Finally, we define \texttt{mapL} for inductive lists in the usual way,
by applying the constructor helper function \texttt{mkList} to the
intersection pair of the Church and parametricity components, and
the reflection component:

\begin{verbatim}
mapL ◂ ∀ A : ★ . ∀ B : ★ .  Π f : A ➔ B . List · A ➔ List · B
  = Λ A . Λ B . λ f . λ xs . mkList · B
  [ mapCL · A · B f xs , mapPL · A · B f xs ]
  -(mapRL · A · B f xs) .
\end{verbatim}

Because \texttt{mkList} is an identity coercion, as are the components
\texttt{mapCL} and \texttt{mapPL}, the partial application of
\texttt{mapL} to the identity function is also an identity coercion:

\begin{theorem}
$\erase{\name{mapL}} \earg (\fun{x}{x})$ is the identity function:
\labthm{mapl}
\end{theorem}

\begin{proof}
{\small
\begin{align*}
  &=_\delta (\fun{f} \fun{xs} \erase{
    \name{mkList} \earg
    \pair{
      \name{mapCL} \earg f \earg xs
    }{
      \name{mapPL} \earg f \earg xs
    }
  }) \earg (\fun{x}{x})
  && \by{Reduce.}
  \\
  &=_\beta \fun{xs} \erase{
    \name{mkList} \earg
    \pair{
      \Erase{\name{mapCL}} \earg (\fun{x}{x}) \earg xs
    }{
      \name{mapPL} \earg (\fun{x}{x}) \earg xs
    }
  }
  && \by{By \reflem{mapcl}.}
  \\
  &=_\beta \fun{xs} \erase{
    \name{mkList} \earg
    \pair{
      xs
    }{
      \Erase{\name{mapPL}} \earg (\fun{x}{x}) \earg xs
    }
  }
  && \by{By \reflem{mappl}.}
  \\
  &=_\beta \fun{xs} \erase{
    \name{mkList} \earg
    \pair{
      xs
    }{
      xs
    }
  }
  && \by{Erase pair.}
  \\
  &=~ \fun{xs} \erase{
    \name{mkList} \earg xs
  }
  && \by{By \reflem{mkvec}.}
  \\
  &=_\beta \fun{xs}{xs}
  \eqed
\end{align*}}
\end{proof}


\subsection{Nested Identity Coercions}

In order to reuse a program over nested lists to derive a program
over nested vectors, we must coerce the nested vectors input of the
derived program to nested lists. Below, we define \texttt{v2l-v2l}
to perform such a coercion between nested datatypes.

\begin{verbatim}
v2l-v2l ◂ ∀ A : ★ . ∀ n : Nat . ∀ m : Nat .
  Vec · (Vec · A n) m ➔ List · (List · A)
  = Λ A . Λ n . Λ m . λ xss . mapL · (Vec · A n) · (List · A) 
   (v2l · A -n) (v2l · (Vec · A n) -m xss) .
\end{verbatim}

Unsurprisingly, we define \texttt{v2l-v2l} by mapping \texttt{v2l}
(using \texttt{mapL}) over the result of coercing the outer vector to a
list (again via \texttt{v2l}). However, we now have an instance of a
\texttt{mapL} applied to an identity coercion (\texttt{v2l}), which
allows \texttt{v2l-v2l} to be an identity coercion between nested types:

\begin{theorem}
\erase{\name{v2l-v2l}} is the identity function:
\labthm{v2lv2l}
\end{theorem}

\begin{proof}
{\small
\abovedisplayskip=-\baselineskip
\begin{align*}
  &=_\delta \fun{xss} \erase{
    \name{mapL} \earg \Erase{\name{v2l}} \earg (\name{v2l} \earg xss)
  }
  && \by{By \refthm{v2l}.}
  \\
  &=~ \fun{xss} \erase{
    \Erase{\name{mapL}} \earg (\fun{x}{x}) \earg (\name{v2l} \earg xss)
  }
  && \by{By \refthm{mapl}.}
  \\
  &=_\beta \fun{xss} \erase{
    \name{v2l} \earg xss
  }
  && \by{By \refthm{v2l}.}
  \\
  &=_\beta \fun{xss}{xss}
  \eqed
\end{align*}}
\end{proof}

Because we plan on reusing an unindexed function over nested lists to
define an indexed function over nested vectors, we will need to
remember the length constraints on coerced nested lists
(like in \refsec{listreuse}). Thus, we also define the
nested mapping function \texttt{v2u-v2l}, which maps the outer vector
to a list, but remembers the inner list length constraints by mapping
the inner vectors to length-constrained lists (\texttt{VecL}):

\begin{verbatim}
v2u-v2l ◂ ∀ A : ★ . ∀ n : Nat . ∀ m : Nat .
  Vec · (Vec · A n) m ➔ List · (VecL · A n)
  = Λ A . Λ n . Λ m . λ xss .
  (mapL · (Vec · A n) · (VecL · A n) (v2u · A -n)
    (v2l · (Vec · A n) -m xss)) .
\end{verbatim}

Just like \texttt{v2l-v2l}, \texttt{v2u-v2l} is also an identity
coercion:

\begin{theorem}
\erase{\name{v2u-v2l}} is the identity function:
\labthm{v2uv2l}
\end{theorem}

\begin{proof}
{\small
  Same as the proof of \refthm{v2lv2l}, but erasing
  \name{v2u} (instead of \name{v2l})
  by \refthm{v2u} in the first step. \qed
}
\end{proof}

\subsection{Identity Coercion from VecL to List}

If we have a length-constrained list (\texttt{VecL}), we can retrieve
the inner list as the first projection of intersection:

\begin{verbatim}
u2l ◂ ∀ A : ★ . ∀ n : Nat . VecL · A n ➔ List · A
  = Λ A . Λ n . λ xs . xs.1 .
\end{verbatim}

The nice thing about length-constrained lists is that they erase to
their list component, preventing the constraint from interfering with
definitional equalities. Similarly, the projection of the list from
the length-constrained list is an identity coercion, preventing the
constraint from incurring runtime overhead:

\begin{lemma}
\erase{\name{u2l}} is the identity function:
\lablem{u2l}
\end{lemma}

\begin{proof}
{\small
Erase projection. \qed
}
\end{proof}

We can use \texttt{v2u-v2l} to coerce a vector of vectors to a list of
length-constrained lists, allowing us to rewrite by the constraints to
prove that derived vector programs have the appropriate
indices. However, ultimately we want to reuse a nested list program,
so we also a define \texttt{u2l-l2l} to project away the constraints
of the inner length-constrained lists:

\begin{verbatim}
u2l-l2l ◂ ∀ A : ★ . ∀ n : Nat .
  List · (VecL · A n) ➔ List · (List · A)
  = Λ A . Λ n . λ xss .
  mapL · (VecL · A n) · (List · A) (u2l · A -n) xss .
\end{verbatim}

Once again, this results in a nested identity coercion:

\begin{lemma}
\erase{\name{u2l-l2l}} is the identity function:
\lablem{u2ll2l}
\end{lemma}

\begin{proof}
{\small
\abovedisplayskip=-\baselineskip
\begin{align*}
  &=_\delta \fun{xss} \erase{
    \name{mapL} \earg \Erase{\name{u2l}} \earg xss
  }
  && \by{By \reflem{u2l}.}
  \\
  &=~ \fun{xss} \erase{
    \name{mapL} \earg (\fun{x}{x}) \earg xss
  }
  && \by{By \refthm{mapl}.}
  \\
  &=_\beta \fun{xss}{xss}
  \eqed
\end{align*}}
\end{proof}

\subsection{Program Reuse}
\labsec{nestreuse:progreuse}

\begin{figure}[t]
\centering
\begin{verbatim}
lengthDistConcat ◂ ∀ A : ★ . ∀ n : Nat . Π xss : List · (VecL · A n) .
  mult (length · (VecL · A n) xss) n ≃ 
  length (concatL · A (u2l-l2l · A -n xss))
\end{verbatim}
\caption{Length distributes through concat.}
\labfig{lendistconcat}
\end{figure}

Now we reuse a list concatenation program
(\texttt{concatL}, flattening a list of lists to a list)
to derive a vector concatenation program (\texttt{concatV}).
Once again, we coerce the result of reusing
\texttt{concatL} to a vector (using \texttt{l2v}). However,
this time we reuse \texttt{concatL} by applying it to the result of
mapping the input vector of vectors to a list of lists
(via \texttt{v2l-v2l}). Vector concatenation (\texttt{concatV})
requires the index of the resulting vector to equal the product of the
outer and inner input vector lengths (\texttt{mult m n}), thus we must
also perform rewrites to ensure that our reused list program
(\texttt{concatL}) respects this indexing requirement.

\begin{verbatim}
concatV ◂ ∀ A : ★ . ∀ n : Nat . ∀ m : Nat .
  Vec · (Vec · A n) m ➔ Vec · A (mult m n)
  = Λ A . Λ n . Λ m . λ xss .
  ρ (v2u · (Vec · A n) -m xss).2 -
  ρ (lengthDistConcat · A -n (v2u-v2l · A -n -m xss)) -
  (l2v · A (concatL · A (v2l-v2l · A -n -m xss))) .
\end{verbatim}

The result of reusing \texttt{concatL} has the following type:

\begin{verbatim}
Vec · A (length (l2v (concatL (v2l-v2l -n -m xss))))
\end{verbatim}

We rewrite by the property
(\texttt{lengthDistConcat} in \reffig{lendistconcat})
that length distributes through the list concatenation of
the nested coercion (performed by \texttt{u2l-l2l},
which takes a list of $n$-length-constrained lists to a
list of lists). The result of this distribution is the product of the
length of the nested list and $n$:

\begin{verbatim}
Vec · A (mult (length xss) n)
\end{verbatim}

Note that the property \texttt{lengthDistConcat} relies on all nested
lists having the same length (\texttt{n}), hence it is defined for a
list of length-constrained lists. Yet, our type resulting from reusing
\texttt{concatL} applies \texttt{length} to a list of
(non-constrained) lists \texttt{(l2v (concatL (v2l-v2l -n -m xss)))},
so why does the rewrite using \texttt{lengthDistConcat} succeed?
The reason is that both \texttt{(concatL · A (u2l-l2l · A -n xss))}
and \texttt{(l2v (concatL (v2l-v2l -n -m xss)))} erase to
\texttt{(concatL xss)}! The former is a consequence of
identity coercion \texttt{u2l-l2l} (\reflem{u2ll2l}),
and the latter is a consequence of identity coercions
\texttt{l2v} (\refthm{l2v}) and \texttt{v2l-v2l} (\refthm{v2lv2l}).

Finally, we rewrite by the length constraint on the length of the
outer input vector (using projection \texttt{(v2u xss).2}),
changing \texttt{(length xss)} to \texttt{m}, resulting
in our goal type.

\begin{remark}
The rich definitional equalities introduced by erasure and identity
coercions make program reuse of \texttt{concatV} in terms of
\texttt{concatL} straightforward, allowing us to easily rewrite
our goal type by \texttt{lengthDistConcat}. Program reuse in a
non-erased setting requires more complex lemmas and rewrites, due to
\texttt{VecL} being defined as a non-erased dependent pair
($\Sigma$-type), rather than an erased dependent intersection
($\iota$-type).
\end{remark}

\begin{theorem}
\erase{\name{concatV}} is \erase{\name{concatL}}:
\labthm{concatv}
\end{theorem}

\begin{proof}
{\small
  Erase rewrites, then:
\begin{align*}
  &=_\delta \fun{xss} \erase{
    \name{l2v} \earg (\name{concatL} \earg \Erase{\name{v2l-v2l} \earg xss})
  }
  && \by{By \refthm{v2lv2l}.}
  \\
  &=_\beta \fun{xss} \erase{
    \name{l2v} \earg (\name{concatL} \earg xss)
  }
  && \by{By \refthm{l2v}.}
  \\
  &=_\beta \fun{xss} \erase{\name{concatL}} \earg xss
  && \by{Contract.}
  \\
  &=_\eta \erase{\name{concatL}}
  \eqed
\end{align*}
}
\end{proof}

\subsection{Proof Reuse}
\labsec{nestreuse:proofreuse}

We achieve proof reuse by proving that vector concat distributes
through vector append (\texttt{concatDistAppendV})
in terms of the corresponding proof for lists
(\texttt{concatDistAppendL}). This only requires applying our reused
proof of \texttt{concatDistAppendL} to the result of coercing
our input nested vector arguments to nested lists
(via \texttt{v2l-v2l}):

\begin{verbatim}
concatDistAppendV ◂ ∀ A : ★ .
  ∀ n1 : Nat . ∀ m1 : Nat . Π xss : Vec · (Vec · A n1) m1 .
  ∀ n2 : Nat . ∀ m2 : Nat . Π yss : Vec · (Vec · A n2) m2 .
  appendV (concatV xss) (concatV yss) ≃ concatV (appendV xss yss)
  = Λ A . Λ n1 . Λ m1 . λ xss . Λ n2 . Λ m2 . λ yss .
  concatDistAppendL · A
    (v2l-v2l · A -n1 -m1 xss)
    (v2l-v2l · A -n2 -m2 yss) .
\end{verbatim}

The result of reusing \texttt{concatDistAppendL} has the following type:
\begin{verbatim}
appendL (concatL (v2l-v2l xss)) (concatL (v2l-v2l yss)) ≃
  concatL (appendL (v2l-v2l xss) (v2l-v2l yss))
\end{verbatim}

After erasure, this $\beta$-reduces to our goal because
\texttt{appendL} erases to \texttt{appendV} by
\refthm{appendv},
\texttt{concatL} erases to \texttt{concatV} by
\refthm{concatv},
and \texttt{v2l-v2l} erases to (\texttt{λ x . x})
by \refthm{v2lv2l}.

\section{Related Work}
\labsec{related}

\subsection{Coercible in Haskell}

Breitner et al. describe a GHC extension to Haskell (available
starting with GHC 7.8) for a type class \verb|Coercible a b|, which
allows casting from \verb|a| to \verb|b| when such a cast is indeed
the identity function~\cite{breitner+16}.  The motivation is to support retyping of data
defined using Haskell's \verb|newtype| statement, which is designed to
give programmers the power to erect abstraction barriers that cannot
be crossed outside of the module defining the \verb|newtype|.  Within
such a module, however, \verb|Coercible a b| and associated cast
function \verb|coerce : a -> b| allow programmers to apply zero-cost
casts to change between a \verb|newtype| and its definition.

\verb|Coercible| had to be added as primitive to GHC, along with a
rather complex system of \emph{roles} specifying how coercibility of
application of type constructors follows from coercibility of
arguments to those constructors.  In contrast, in the present work, we
have shown how to derive zero-cost coercions within the existing type
theory of Cedille, with no extensions.  On the other hand, much of the
complexity of \verb|Coercible| in GHC arises from (1) how it
interoperates with programmer-specified abstraction (via
\verb|newtype|) and (2) the need to resolve \verb|Coercible a b| class
constraints automatically, similarly to other class constraints in
Haskell.  The present work does not address either issue.
However, the present work does allow for dependent casts between
indexed variants of datatypes, which \verb|Coercible| does not cover.

\subsection{Ornaments}

Ornaments~\cite{ornaments:original} are used to define refined
version of types (e.g. \texttt{Vec}) from unrefined types
(e.g. \texttt{List}) by ``ornamenting'' the unrefined type with extra
index information. In contrast, our work establishes a relationship
between \texttt{Vec} and \texttt{List} after-the-fact, by defining
identity coercions in both directions for existing types.
By \textit{defining} vectors as natural-number-ornamented lists,
ornaments can be used to calculate the ``patch'' type necessary to adapt a
function from one type to another type~\cite{ornaments:functional}.
For example, ornaments could
calculate that \texttt{lengthDistAppend} is necessary to adapt
\texttt{appendL} from lists to vectors (\texttt{appendV}).

Although ornaments can be used to derive coercions between
types in an ornamental
relationship~\cite{ornaments:original,ornaments:relational},
they will not be identity coercions.
Besides refining the indices of existing datatypes, ornaments
also allow data to be added to existing datatypes. For example, vectors can
be index-refined lists, but lists can also be natural numbers with
elements added. Our work only covers the index refinement aspect of
ornaments.

\subsection{Type Theory in Color}

Type Theory in Color (TTC)~\cite{bernardy:color}
generalizes the concept of erased arguments
of types to various colors, which may be erased optionally and
independently according to modalities in the type theory. In the vector
datatype declaration, the index data can be colored. If a vector is
passed to a function expecting a list (whose modality enforces the
lack of the index data color), then a free non-dependent identity
coercion (using our parlance) is performed.

Lists can also be used as vectors, via a free dependent identity
coercion in the other direction. This works due to a mechanism to
interpret lists as a predicate on natural numbers. The list predicate is
generated as the erasure of its colored elements (like ornaments,
colors can add data in addition to refining indices), which results in
refining lists by the length \textit{function}.

Our work can be used to define a dependent identity coercion
from natural numbers to the datatype of
finite sets (\texttt{Fin}). This is not possible with colors, because
\texttt{Fin} is indexed by successor (\texttt{suc}) in both of its
constructors, which would require generating a predicate on the
natural numbers from a non-deterministic function (or
\textit{relation}). Colors allow identity coercions to be
generated and \textit{implicitly} applied because colors erase \textit{types}, as
well as values, whereas implicit products only erase values
(e.g. $\Lambda$ is erased, but not $\forall$).
Thus, while identity coercions
need to be \textit{explicitly} crafted and applied in our setting, we
are able to \textit{define} identity coercions (like taking natural numbers to
finite sets) for which there is no unique solution.

\section{Conclusion}
\labsec{conc}

We have demonstrated how to achieve zero-cost program and proof reuse
between lists and vectors, which scales to the nested datatype
setting, through the use of \textit{identity coercions}, which erase to
the identity function. Our technique works for datatypes like lists
and vectors, where vectors are the length-indexed version of
lists. Vectors have a subtype relationship with lists, and vice versa,
supporting identity coercion in both directions.

For future work, we
would like to explore what sort of program and proof reuse is possible
(via identity coercions) between types that only have a subtyping
relationship in one direction, such as untyped and intrinsically typed
versions of $\lambda$-calculus expressions.
We would also like to explore integrating a notion of ornaments
into our setting, to automate the generation of the ``patch'' types
necessary for program reuse.
Finally, we would like to generalize our results to a class of
datatypes related by refinement, via a generic encoding of indexed
datatypes.

\bibliographystyle{splncs03}
\bibliography{proof-reuse}


\end{document}